\documentclass[a4paper,UKenglish]{lipics}

\usepackage{boxedminipage}
\usepackage{amsmath}
\usepackage{amssymb}
\usepackage{algpseudocode}

\DeclareMathOperator*{\argmax}{arg\,max}
\newcommand{\hstar}{H^{*}}
\newcommand{\dstar}{d^{*}}
\newcommand{\dellhstar}{D_\ell \cap \hstar}
\newcommand{\crd}[1]{|#1|}
\newcommand{\calM}{{\cal M}}
\newcommand{\calI}{{\cal I}}
\newcommand{\denm}{\tt{DEN-M}}
\newcommand{\depden}{\mathtt{DEN{dep}}}
\newcommand{\rplus}{\mathbb{R}^{+}}

\newcommand{\natnum}{\mathbb{N}}
\newcommand{\geqs}{\geqslant}
\newcommand{\leqs}{\leqslant}

\newtheorem{claim}{Claim}[]
\newcommand{\eat}[1] {}

\newcommand{\cL}{{\cal L}}

\newcommand{\eps}{\epsilon}

\newcommand{\lp}{\mathrm{LP}}
\newcommand{\ip}{\mathrm{IP}}

\newcommand{\cp}{\mathrm{CP}}
\newcommand{\bfs}{\mathrm{BFS}}

\newcommand{\np}{\mathrm{NP}}

\newcommand{\val}{\mathrm{VAL}}

\newcommand{\tum}{{\mathrm{TUM}}}

\eat{
\title{Density Functions subject to a Co-Matroid Constraint}
\runningtitle{Density Functions}
\runningauthors{V. T. Chakaravarthy et.al}
\author{Venkatesan T. Chakaravarthy\inst{1}, Natwar Modani\inst{1}, Sivaramakrishnan R Natarajan\inst{2}, Sambuddha Roy\inst{1}, Yogish Sabharwal\inst{1}}

\institute{1}{
IBM Research -- India\\
New Delhi
}
\affiliation{
  \email{\{vechakra,namodani,sambuddha,ysabharwal\}@in.ibm.com}
}  
\institute{2}{
IIT Madras \\
Chennai, India
}
\affiliation{
\email{sivaramakrishnan.n.r@gmail.com}
}
}

\title{Density Functions subject to a Co-Matroid Constraint\footnote{Work done by the third author while he was interning at IBM Research}}
\titlerunning{Density Functions} 

\author[1]{Venkatesan T. Chakaravarthy}
\author[1]{Natwar Modani}
\author[2]{Sivaramakrishnan R. Natarajan}
\author[1]{Sambuddha Roy}
\author[1]{Yogish Sabharwal}
\affil[1]{IBM Research, New Delhi, India\\
  \texttt{\{vechakra,namodani,sambuddha,ysabharwal\}@in.ibm.com}}
\affil[2]{IIT Chennai, India\\
  \texttt{sivaramakrishnan.n.r@gmail.com}}
\authorrunning{V. Chakaravarthy et. al} 

\begin{document}

\maketitle
\begin{abstract}
In this paper we consider the problem of finding the {\em densest} subset subject to {\em co-matroid constraints}. We are given a {\em monotone supermodular} set function $f$ defined over a universe $U$, and the density of a subset $S$ is defined to be $f(S)/\crd{S}$. This generalizes the concept of graph density. Co-matroid constraints are the following: given matroid $\calM$ a set $S$ is feasible, iff the complement of $S$ is {\em independent} in the matroid. Under such constraints, the problem becomes $\np$-hard. The specific case of graph density has been considered in literature under specific co-matroid constraints, for example, the cardinality matroid and the partition matroid. We show a $2$-approximation for finding the densest subset subject to co-matroid constraints. Thus, for instance, we improve the approximation guarantees for the result for partition matroids in the literature.
\end{abstract}

\section{Introduction}

In this paper, we consider the problem of computing the densest subset with respect to a {\em monotone supermodular} function 
subject to {\em co-matroid} constraints.
Given a universe $U$ of $n$ elements, a function $f: 2^{U} \rightarrow \rplus$ is {\em supermodular} iff 
\[
f(A) + f(B) \leqslant f(A\cup B) + f(A \cap B)
\]
for all $A, B \subseteq U$. If the sign of the inequality is reversed for all $A, B$, then we call the function 
{\em submodular}.
The function $f$ is said to be monotone 
 if $f(A) \leqslant f(B)$ whenever $A\subseteq B$; we assume $f(\emptyset) = 0$. 
We define a {\em density function} $d: 2^{U} \rightarrow \rplus$ as $d(S) \triangleq f(S)/\crd{S}$.
Consider the problem of maximizing the density function $d(S)$ given oracle access to the function $f$.
We observe that the above problem can be solved in polynomial time (see Theorem~\ref{BBB}).

The main problem considered in this paper is to maximize $d(S)$ subject to  certain constraints
that we call {\em co-matroid} constraints.
In this scenario, 
we are given a {\em matroid} $\calM = (U, \calI)$ where $\calI\subseteq 2^{U}$ is the family of {\em independent} sets (we give the formal definition of a matroid in Section~\ref{defs}).
A set $S$ is considered feasible iff  the complement of $S$ is {\em independent} i.e. $\overline{S} \in \calI$.  
The problem is to find the densest feasible subset $S$ given oracle access to $f$ and $\calM$. 
We denote this problem as {\denm}.

We note that even special cases of the {\denm} problem are $\np$-hard~\cite{ks09}.
The main result in this paper is the following:
\begin{theorem}
\label{main.thm}
Given a monotone supermodular function $f$ over a universe $U$, and a
matroid ${\cal M}$ defined over the same universe, there is a $2$-approximation 
algorithm for the {\denm} problem.
\end{theorem}

Alternatively one could have considered the same problem under {\em matroid constraints} (instead of co-matroid constraints). 
We note that this problem is significantly harder,
since the Densest Subgraph problem can be reduced to special cases of this problem (see \cite{ac09,ks09}).
The Densest Subgraph problem is notoriously hard: 
the best factor approximation known to date is 
 $O(n^{1/4 + \eps})$ for any $\eps > 0$~\cite{feige10}.

Special cases of the {\denm} problem have been extensively studied in the context of graph density, 
and we discuss this next. 

\subsection{Comparison to Graph Density}

Given an undirected graph $G = (V, E)$, the density  $d(S)$ of a subgraph on vertex set $S$ is defined
as the quantity $\frac{|E(S)|}{|S|}$, where $E(S)$ is the set of edges in the subgraph induced by the 
vertex set $S$. The densest subgraph problem is to find the subgraph $S$ of $G$ that maximizes the 
density. 

The concept of graph density is ubiquitous, more so in the context of social networks. 
In the context of social networks, the problem is to detect {\em communities}: collections of individuals who are relatively well connected as compared to other parts of the social network graph. 

The results relating to graph density have been fruitfully applied to finding communities in the social network graph (or even web graphs, gene annotation graphs~\cite{ks10}, problems related to the formation of most effective teams~\cite{gajewar-sarma}, etc.).
Also, note 
that graph density appears naturally in the study of threshold phenomena in random graphs, see \cite{alon-spencer}. 

Motivated by applications in social networks, the graph density problem and its variants have been well studied.
Goldberg~\cite{goldberg84}  proved that the densest subgraph problem  can be solved optimally
in polynomial time: he showed this via a reduction to a series of max-flow computations.
Later, others~\cite{charikar00, ks09} have given new proofs for the above result, motivated by considerations to extend the result to some generalizations and variants.

Andersen and Chellapilla~\cite{ac09} studied the following generalization of the above problem.
Here, the input also includes an integer $k$, and the goal is to find the densest subgraph $S$ 
subject to the constraint $\crd{S} \geqslant k$. This corresponds to finding {\em sufficiently large}
dense subgraphs in social networks. This problem is $\np$-hard~\cite{ks09}. Andersen and Chellapilla~\cite{ac09}
gave a $2$-approximation algorithm. Khuller and Saha~\cite{ks09} give two alternative  
algorithms: one of them is a greedy procedure, while the other is $\lp$-based. Both the algorithms have
$2$-factor guarantees.

Gajewar and Sarma~\cite{gajewar-sarma} consider a further generalization. The input also includes a partition 
of the vertex set into $U_1, U_2, \cdots, U_t$, and non-negative integers $r_1, r_2, \cdots, r_t$. The goal is 
to find the densest subgraph $S$ subject to the constraint that for all $1\leqs i \leqs t$, $\crd{S\cap U_i} \geqslant r_i$. 
They gave a $3$-approximation algorithm by extending the greedy procedure of Khuller and Saha~\cite{ks09}.

We make the following observations: 
(i) The objective function $\crd{E(S)}$ is monotone and supermodular.
(ii) The constraint $\crd{S} \geqslant k$ (considered by \cite{ac09}) is a co-matroid constraint; this corresponds to 
the {\em cardinality matroid}.
(iii) The constraint considered by Gajewar and Sarma~\cite{gajewar-sarma} is also a co-matroid constraint; this
corresponds to the {\em partition matroid}
(formal definitions are provided in Section~\ref{defs}).
Consequently, our main result Theorem~\ref{main.thm} improves upon the above results in three directions:
\begin{itemize}
\item {\em Objective function}: Our results apply to general monotone supermodular functions $f$ instead of the specific 
set function $\crd{E(S)}$ in graphs.
\item {\em Constraints}: We allow co-matroid constraints corresponding to {\em arbitrary} matroids. 
\item {\em Approximation Factor:} For the problem considered by Gajewar and Sarma~\cite{gajewar-sarma}, we improve the approximation guarantee from $3$ to $2$.
We match the best factor known for the 
at-least-$k$ densest subgraph problem considered in \cite{ac09, ks09}.
\end{itemize}

\subsection{Other Results}

\noindent
{\it Knapsack Covering Constraints:}

We also consider the following variant of the {\denm} problem.
In this variant, we will have a weight $w_i$ (for $i = 1,\cdots, \crd{U}$) for every element $i \in U$, and a number $k \in \natnum$.
A set $S$ of elements is {\em feasible} if and only if the following condition holds:
\[
\sum_{i\in S} w_i \geqslant k
\]
We call this a {\em knapsack covering constraint}.
We extend the proof of Theorem~\ref{main.thm} to show the following:
\begin{theorem}
\label{knapsack}
Suppose we are given a monotone supermodular function $f$ over a universe $U$, weights $w_i$ for every element $i\in U$, and a number $k\in \natnum$.  
Then there is a $3$-approximation algorithm for maximizing the {\em density function} 
$d(S)$ subject to {\em knapsack covering constraints} corresponding to the weights $w_i$ and the number $k$.
\end{theorem}

\noindent
{\it Dependency Constraints:}

Saha et.~al\cite{ks10} consider a variant of the graph density problem.
In this version, we are given a specific collection of vertices $A \subseteq V$; a subset $S$ of vertices is {\em feasible} iff $A \subseteq S$. We call this restriction the {\em subset} constraint. The objective is to find the densest subgraph among subsets satisfying a subset constraint. Saha et.~al\cite{ks10} prove that this problem is solvable in polynomial time by reducing this problem to a series of max-flow computations.

We  study a generalization of the subset constraint problem. Here, we are given a monotone supermodular function
$f$ defined over universe $U$. 
Additionally, we are given a 
{\em directed graph} $D = (U, \vec{A})$ over the universe $U$. A feasible solution $S$ has to satisfy the following property: if $a \in S$, then every vertex of the digraph $D$ {\em reachable} from $a$ also has to belong to $S$. Alternatively, $a \in S$ and $(a, b) \in \vec{A}$ implies that $b \in S$.  We call the digraph $D$ as the {\em dependency graph} and such constraints as {\em dependency} constraints.
The goal is to find the densest subset $S$ subject to the dependency constraints. We call this the $\depden$ problem.
We note that the concept of dependency constraints generalizes that of the subset constraints: construct a digraph $D$ by drawing directed arcs from every vertex in $U$ to every vertex in $A$.  
The motivation for this problem comes from certain considerations in 
social networks, where we are to find the densest subgraph but with the restriction that in the solution subgraph all the members of a sub-community (say, a family) are present or absent simultaneously. 
In literature, such a solution $S$ that satisfies the dependency constraints is also called a {\em closure} (see~\cite{topkis-book}, Section 3.7.2). Thus our problem can be rephrased as that of finding the densest subset over all closures. 

We note that dependency constraints are incomparable with co-matroid constraints. In fact dependency constraints are not even upward monotone: it is {\em not} true that 
if $S$ is a feasible subset, {\em any} superset of $S$ is feasible.

Our result is as follows:
\begin{theorem}
\label{depden}
The $\depden$ problem is solvable in polynomial time.
\end{theorem}

The salient features of the above result are as follows:
\begin{itemize}
\item While the result in \cite{ks10} is specific to graph density, our result holds for density functions arising from 
arbitrary monotone supermodular functions. 
\item Our proof of this result is $\lp$-based. The work of \cite{ks10} is based on max-flow computations. 
We can 
extend our $\lp$-based approach (via convex programs) to the case for density functions arising from arbitrary 
monotone supermodular $f$, while we are not aware as to how to extend the max-flow based computation.
\item The proof technique, inspired by Iwata and Nagano~\cite{iwata-nagano09} also extends to show ``small support''
results: thus, for instance, we can show that for the $\lp$ considered by \cite{ks09} for the at-least-k-densest subgraph problem, 
every {\em non-zero} component of any basic feasible solution is one of {\em two} values.
\end{itemize}

\noindent
{\it Combination of Constraints:}

We also explore the problem of finding the densest subset subject to a combination of the constraints considered. We are able to prove results for the problem of maximizing a density function 
subject to (a) {\em co-matroid} constraints and (b) {\em subset} constraints. 
Suppose we are given a monotone supermodular function $f$ over a universe $U$, a matroid $\calM = (U, \calI)$, and a subset of elements $A\subseteq U$.
A subset $S$ is called feasible iff (1) $S$ satisfies the co-matroid constraints wrt $\calM$ (i.e. $\overline{S} \in \calI$) and (2) $S$ satisfies the subset constraint
wrt $A$ (i.e. $A \subseteq S$). We show the following:
\begin{theorem}
\label{combo}
There is a $2$-approximation algorithm for the problem of maximizing the {\em density function} $d(S)$ corresponding 
to a monotone supermodular function $f$, subject to the co-matroid and subset constraints.
\end{theorem}

\subsection{Related Work}

Recently, there has been a considerable interest in the problems of optimizing submodular functions under various types of constraints. The most common constraints that are considered are {\em matroid constraints}, {\em knapsack constraints} or combinations of the two varieties. Thus for instance, Calinescu et. al~\cite{calinescu11} considered the problem of maximizing a {\em monotone} submodular function subject to a matroid constraint. They provide an algorithm and show that it yields a $(1 - 1/e)$-approximation: this result is essentially optimal (also see the recent paper \cite{fw12} for a combinatorial algorithm for the same). Goemans and Soto~\cite{goemans-soto} consider the problem of minimizing a {\em symmetric} submodular function subject to arbitrary {\em matroid} constraints. They prove the surprising result that this problem can be solved in polynomial time. In fact, their result extends to the significantly more general case of {\em hereditary constraints}; the problem of extending our results to {\em arbitrary} hereditary functions is left open.

The density functions that we consider may be considered as ``close'' to the notion of supermodular functions. To the best of our knowledge, the general question of {\em maximizing} density functions subject to a (co-)matroid constraint has never been considered before. 

\subsection{Proof Techniques}

We employ a greedy algorithm to prove Theorems~\ref{main.thm} and ~\ref{knapsack}. Khuller and Saha~\cite{ks09}
and Gajewar and Sarma~\cite{gajewar-sarma} had considered a natural greedy algorithm 
for the problem of maximizing graph density subject to co-matroid constraints corresponding to the 
cardinality matroid and partition matroid respectively. 
Our greedy algorithm can be viewed as an abstraction of the natural greedy algorithm to the 
generalized scenario of arbitrary monotone supermodular functions.
However, our analysis is different from that in \cite{ks09, gajewar-sarma}: the efficacy of 
our analysis is reflected in the fact that we improve on the guarantees provided by \cite{gajewar-sarma}.
While they provide
a $3$-approximation algorithm for the graph density problem with partition matroid 
constraints, we use the modified analysis to obtain a $2$-factor guarantee.
In both of the earlier papers \cite{ks09, gajewar-sarma}, a particular stopping condition is employed to 
define a set $D_\ell$ useful in the analysis. For instance, in Section 4.1 of \cite{gajewar-sarma} they 
define $D_\ell$ using the optimal set $\hstar$ directly. We choose a different {\em stopping condition}
to define the set $D_\ell$; it turns out that this choice is crucial for achieving a $2$-factor guarantee.

We prove Theorem~\ref{depden} using $\lp$-based techniques. 
In fact, we  provide two proofs for the same.  Both our techniques also provide alternate new proofs 
of the basic result that graph density is computable in polynomial time. The first proof method is inspired by
Iwata and Nagano~\cite{iwata-nagano09}. The second proof method invokes Cramer's rule to 
derive the conclusion.

\subsection{Organization}
We present the relevant definitions in Section~\ref{defs}. We proceed to give the proof of Theorem~\ref{main.thm} in 
Section~\ref{proofmain}, while the proof of Theorem~\ref{knapsack} is presented in 
Section~\ref{app:knapsack}. 
The proof of Theorem~\ref{depden} is presented in Section~\ref{sec:depden}.
and the proof of Theorem~\ref{combo} is in Section~\ref{app:combo}.
\section{Preliminaries}
\label{defs}
In this paper, we will use the following notation: given {\em disjoint} sets $A$ and $B$ we will use 
$A+B$ to serve as shorthand for $A\cup B$. Vice versa, when we write $A +B$ it will hold implicitly that the sets $A$ and $B$ are disjoint.

\noindent
{\bf Monotone:}
A set function $f$ is called {\em monotone} if $f(S) \leqslant f(T)$ whenever $S \subseteq T$. 

\noindent
{\bf Supermodular:}
A set function $f: 2^{U} \rightarrow \rplus$ over a universe $U$ is called {\em supermodular} if the following holds for any two sets 
$A, B \subseteq U$:
\[
f(A) + f(B) \leqslant f(A\cup B) + f(A \cap B)
\]

If the inequality holds (for every $A, B$) with the sign reversed, then the function $f$ is called {\em submodular}. 
In this paper, we will use the following equivalent definition of supermodularity:
given disjoint sets $A, B$ and $C$, 
\[
f(A + C) - f(A) \leqslant f(A  + B + C) - f(A + B)
\]
We can think of this as follows: the {\em marginal utility} of the set of elements $C$ to the set $A$ increases as the set becomes ``larger" ($A+B$ instead of $A$).

It is well known (see \cite{gls81, schrijver00}) that supermodular functions can be {\em maximized} in polynomial time (whereas submodular functions can be minimized in polynomial time). Let us record this as:
\begin{theorem}
\label{CCC}
Any supermodular function $f: 2^{U} \rightarrow \rplus$ can be maximized in polynomial time.
\end{theorem} 

We also state the following folklore corollary:
\begin{corollary}
\label{BBB}
Given any supermodular function $f: 2^{U} \rightarrow \rplus$, we can find 
$\max_{S} \frac{f(S)}{\crd{S}}$ in polynomial time.
\end{corollary}
For completeness, a proof of this Corollary is included in  Section~\ref{cor.proof}.

\noindent
{\bf Density Function:}
Given a function $f$ over $U$,
the density of a set $S$ is defined to be $d(S) = \frac{f(S)}{\crd{S}}$.

\noindent
{\bf Matroid:}
A matroid is a pair $\calM = (U, \calI)$ where $\calI \subseteq 2^{U}$, and 
\begin{enumerate}
\item (Hereditary Property) $\forall{B\in \calI}, A \subset B \implies A \in \calI$.
\item (Extension Property) $\forall{A, B \in \calI}: \crd{A} < \crd{B} \implies \exists{x \in B\setminus A}: A + x \in \calI$
\end{enumerate}
Matroids are generalizations of vector spaces in linear algebra and are ubiquitous in 
combinatorial optimization because of their connection with greedy algorithms. Typically the sets
in $\calI$ are called {\em independent} sets, this being an abstraction of linear independence in 
linear algebra. The maximal independent sets in a matroid are called the {\em bases} (again preserving the
terminology from linear algebra). An important fact for matroids is that all bases have equal cardinality -- this is an outcome of the Extension Property of matroids. 

Any matroid is equipped with a {\em rank function} $r : 2^{U} \rightarrow \rplus$. The rank of a subset $S$ is 
defined to be the size of the {\em largest} independent set contained in the subset $S$. By the Extension Property, this is well-defined. See the excellent text by Schrijver~\cite{schrijver-book} for details.

Two commonly encountered matroids are the 
(i) {\em Cardinality Matroid:} 
Given a universe $U$ and $r\in \natnum$, the {\em cardinality matroid} is the matroid $\calM = (U, \calI)$, where a set $A$ is {\em independent} (i.e. belongs to $\calI$) iff $\crd{A} \leqslant r$. 
(ii) {\em Partition Matroid:}
Given a universe $U$ and a partition of $U$ as $U_1, \cdots, U_r$ and non-negative integers $r_1, \cdots, r_t$, the {\em partition matroid} is $\calM = (U, \calI)$, where a set $A$ belongs to $\calI$ iff $\crd{A \cap U_i} \leqslant r_i$ for all $i = 1, 2, \cdots, t$.

\noindent
{\bf Convex Programs:} We will need the definition of a convex program, and that they can be solved to arbitrary precision in polynomial time, via the ellipsoid method(see \cite{gls81}).
We refer the reader to the excellent text~\cite{bv-book}.

\section{Proof of Theorem~\ref{main.thm}}
\label{proofmain}
We first present the algorithm and then its analysis. To get started, we 
describe the intuition behind the algorithm. 

Note that co-matroid constraints are {\em upward monotone}: if a set $S$ is feasible for such
constraints, then any {\em superset} of $S$ is also feasible. Thus, it makes sense to 
find a {\em maximal} subset of $U$ with the maximum density. In the following description of the algorithm, one may note that the sets $D_1, D_2, \cdots, D_i$ are an attempt to find the maximal subset with the 
largest density.  Given this rough outline, the algorithm is presented in Figure~\ref{algo}.

\begin{figure}
\begin{center}
\begin{boxedminipage}{0.8\hsize}
\begin{algorithmic}
\State $i \gets 1$
\State $H_i \gets \argmax_{X} \frac{f(X)}{\crd{X}}$
\State $D_i \gets H_i$
\While{$D_i ~\mathrm{ infeasible}$}
\State $H_{i+1} \gets \argmax_{X: X\cap D_i = \emptyset} \frac{f(D_i + X) - f(D_i)}{\crd{X}}$
\State $D_{i+1} \gets D_i + H_{i+1}$
\State $i \gets i+1$
\EndWhile
\State $L \gets i$
\For{$i = 1 \to L$} 
\State{Add arbitrary vertices to $D_i$ to make it minimal feasible}
\State{Call the result $D'_i$}
\EndFor
\State {Output the subset among the $D'_i$'s with the highest density}
\end{algorithmic}
\end{boxedminipage}
\end{center}
\caption{Main Algorithm}
\label{algo}
\end{figure}

We note that we can find the maximum $\max_{X: X\cap D_i = \emptyset} \frac{f(D_i + X) - f(D_i)}{\crd{X}}$
in polynomial time. This is because the function $f(D_i+X)$ for a fixed $D_i$ is supermodular (and we appeal to Corollary~\ref{BBB}).

Let $\hstar$ denote the optimal solution, i.e. the subset that maximizes the density $d(S)$ subject to the 
co-matroid constraints. Let $\dstar$ denote the optimal density, so that $f(\hstar) = \dstar{\cdot}\crd{\hstar}$.

We can make the following easy claim:
\begin{claim}
\label{DDD}
The subset $D_1$ obeys the inequality $d(D_1) \geqslant \dstar$.
\end{claim}
This is because $D_1$ is the densest subset in the universe $U$, while $\dstar$ is the density of a 
specific subset $\hstar$. 

In the following, we will have occasion to apply the following lemmas. 

\begin{lemma}
\label{lem.1}
Let $a, b, c, d, \theta\in \mathbb{R}^{+}$ be such that the inequalities $\frac{a}{b} \geqslant \theta$  and 
$\frac{c}{d} \geqslant \theta$ hold. Then it is true that $\frac{a + c}{b+ d} \geqslant \theta$.
Thus, if 
$\frac{a}{b} \geqslant \frac{c}{d}$, then 
$\frac{a+c}{b+d} \geqslant \frac{c}{d}$ (by setting $\theta = \frac{c}{d}$).
\end{lemma}

Also, 
\begin{lemma}
\label{lem.2}
Let $a, b, c, d \in \mathbb{R}^{+}$ be real numbers such that 
$\frac{a}{b} \geqslant \frac{c}{d}$  holds.
\begin{itemize}
\item Suppose $a \geqslant c$, $b \geqslant d$. Then 
the inequality 
$\frac{a - c}{b - d} \geqslant \frac{a}{b}$ holds.
\item Suppose $c \geqslant a$, $d \geqslant b$. Then 
the inequality 
$\frac{c}{d} \geqslant \frac{c -a}{d-b}$ holds.
\end{itemize}
\end{lemma}

We make the following claim:
\begin{claim}
\label{XXX}
The sequence of subsets $D_1, D_2, \cdots, D_L$ obeys the following ordering:
\[
\frac{f(D_1)}{\crd{D_1}} \geqslant \frac{f(D_2) - f(D_1)}{\crd{D_2} - \crd{D_1}} \geqslant \cdots \geqslant\frac{f(D_{i+1}) - f(D_i)}{\crd{D_{i+1}} - \crd{D_i}} \geqslant\cdots \geqslant\frac{f(D_L) - f(D_{L-1})}{\crd{D_L} - \crd{D_{L-1}}}
\]
\end{claim}

\begin{proof} 
Consider any term in this sequence, say $\frac{f(D_{i+1}) - f(D_i)}{\crd{D_{i+1}} - \crd{D_i}}$. Note that 
$H_{i+1}$ was chosen as arg max of $\frac{f(D_i + X) - f(D_i)}{\crd{X}}$. Therefore, 
$\max_{X} \frac{f(D_i + X) - f(D_i)}{\crd{X}} = \frac{f(D_{i+1}) - f(D_i)}{\crd{D_{i+1}} - \crd{D_i}}$. Hence this quantity is larger
than $\frac{f(D_{i+2}) - f(D_i)}{\crd{D_{i+2}} - \crd{D_i}}$ (as long as $D_{i+2}$ is well defined). 
Now from the second part of Lemma~\ref{lem.2}, we get that 
\[
\frac{f(D_{i+1}) - f(D_i)}{\crd{D_{i+1}} - \crd{D_i}} \geqslant \frac{f(D_{i+2}) - f(D_i)}{\crd{D_{i+2}} - \crd{D_i}} 
\geqslant \frac{f(D_{i+2}) - f(D_{i+1})}{\crd{D_{i+2}} - \crd{D_{i+1}}}
\]
\end{proof}

Via an application of Lemma~\ref{lem.1}, we then have:
\begin{claim}
\label{YYY}
Given any $i$ ($1 \leqslant i \leqslant L$), the following holds:
\[
\frac{f(D_i)}{\crd{D_i}} \geqslant \frac{f(D_{i}) - f(D_{i-1})}{\crd{D_{i}} - \crd{D_{i-1}}}
\]
\end{claim}
\begin{proof}
We will the prove the statement by induction.

\noindent
{\bf Base Case:} We implicitly assume that $D_0 = \emptyset$, and hence the case for $i = 1$ holds.

\noindent
{\bf Induction Step:}
Assume the statement by induction for $i =k$, and we prove it for $i = k+1$. Thus, by hypothesis we have
\[
\frac{f(D_k)}{\crd{D_k}} \geqslant \frac{f(D_{k}) - f(D_{k-1})}{\crd{D_{k}} - \crd{D_{k-1}}}
\]
Now by Claim~\ref{XXX} we have that 
\[
\frac{f(D_{k}) - f(D_{k-1})}{\crd{D_{k}} - \crd{D_{k-1}}} \geqslant \frac{f(D_{k+1}) - f(D_{k})}{\crd{D_{k+1}} - \crd{D_{k}}}
\]
Thus, 
\[
\frac{f(D_k)}{\crd{D_k}} \geqslant  \frac{f(D_{k+1}) - f(D_{k})}{\crd{D_{k+1}} - \crd{D_{k}}}
\]
Applying Lemma~\ref{lem.1}, we get:
\[
\frac{f(D_{k+1})}{\crd{D_{k+1}}} \geqslant \frac{f(D_{k+1}) - f(D_{k})}{\crd{D_{k+1}} - \crd{D_{k}}}
\]
Thus we have proven the Claim by induction.
\end{proof}

The analysis will be broken up into two parts. We will consider the set $D_\ell$ in the 
sequence $D_1, D_2, \cdots, D_L$ such that 
the following hold:
\[
\frac{f(D_\ell) - f(D_{\ell - 1})}{\crd{D_\ell} - \crd{D_{\ell - 1}}} \geqslant \frac{\dstar}{2}
\]
but 
\[
\frac{f(D_{\ell+1}) - f(D_{\ell})}{\crd{D_{\ell+1}} - \crd{D_{\ell}}} < \frac{\dstar}{2}
\]
Since $d(D_1) \geq \dstar$ by Claim~\ref{DDD}, such an $\ell$ will exist or $\ell = L$. If $\ell = L$, then we have a
feasible solution $D_L$ with the property that $\frac{f(D_L) - f(D_{L-1})}{\crd{D_L} - \crd{D_{L-1}}} \geqslant \frac{\dstar}{2}$. Therefore, by Claim~\ref{YYY} we have that $d(D_L) \geqslant \frac{\dstar}{2}$ and we are done 
in this case.

So we may assume that $\ell < L$ so that $D_\ell$ is {\em not} feasible. 
In this case, we will prove that $D'_\ell$ has the correct density, i.e. 
that $d(D'_\ell) \geqslant \frac{\dstar}{2}$.

To this end, 
we will prove two facts about $D_\ell$ and that will yield the desired result:

\begin{claim}
\label{GGG}
\[
f(D_\ell) - f(\dellhstar) \geqslant \frac{\dstar}{2}(\crd{D_\ell} - \crd{\dellhstar})
\]
\end{claim}
\begin{proof}
Note that $D_\ell = H_1 + H_2 + \cdots + H_\ell$.
For brevity, for $1 \leqslant i \leqslant \ell$, denote $H_i \cap \hstar$ as $A_i$ (thus, $A_i \subseteq H_i$ for every $i$).
Thus, $\dellhstar = A_1 + A_2 + \cdots + A_\ell$.

We will prove the following 
statement by induction on $i$ (for $1 \leqslant i \leqslant \ell$):
\[
f(H_1 + H_2 + \cdots + H_i) - f(A_1 + A_2 + \cdots + A_i) \geqslant \frac{\dstar}{2}(\crd{H_1 + H_2 + \cdots + H_i} - \crd{A_1+ A_2 + \cdots + A_i})
\]

\noindent
{\bf Base Case:} For $i = 1$, we have to prove that:
 
\begin{align*}
\frac{f(H_1) - f(A_1)}{|H_1| - |A_1|} &\geqslant \frac{\dstar}{2}
\end{align*}
Since $H_1$ is the {\em densest} subset, we have 
\[
\frac{f(H_1)}{\crd{H_1}} \geqslant \frac{f(A_1)}{\crd{A_1}}
\]
and we may apply (the first part of) Lemma~\ref{lem.2} to obtain the desired.

\noindent
{\bf Induction Step:} Assume the statement to be true for $i$, 
and we will prove it for $i+1$.

Consider the following chain:
\begin{align*}
&\frac{f(H_1 + \cdots + H_i + H_{i+1}) - f(H_1+ \cdots +H_i)}{\crd{H_{i+1}}} \stackrel{H_{i+1} {\argmax}}{\geqslant} \\
&\frac{f(H_1 + \cdots + H_i + A_{i+1}) - f(H_1 + \cdots +H_i)}{\crd{A_{i+1}}} 
\stackrel{supermodular}{\geqslant} \\
&\frac{f(A_1 + \cdots + A_i + A_{i+1}) - f(A_1 + \cdots +A_i)}{\crd{A_{i+1}}} 
\end{align*}

We would now like to apply Lemma~\ref{lem.2} to the first and last terms in the above chain. 
To this end, let us check the preconditions: 
\begin{align*}
&f(H_1 + \cdots + H_i + H_{i+1}) - f(H_1+ \cdots + H_i) \stackrel{monotone}{\geqslant} f(H_1 + \cdots + H_i + A_{i+1}) - f(H_1 + \cdots + H_i) \\
&\stackrel{supermodular}{\geqslant} f(A_1 + \cdots + A_i + A_{i+1}) - f(A_1 + \cdots + A_i)
\end{align*}

Also, clearly, $\crd{H_{i+1}}\geqslant \crd{A_{i+1}}$. 

Thus,  the preconditions in Lemma~\ref{lem.2} hold, and we have that 
\begin{align*} 
&\frac{f(H_1 + \cdots + H_{i+1}) - f(A_1 + \cdots + A_{i+1}) - f(H_1 + \cdots + H_i) + f(A_1  + \cdots + A_i)}{\crd{H_{i+1}} - \crd{A_{i+1}} } \geqslant \\
&\frac{f(H_1 + \cdots + H_i + H_{i+1}) - f(H_1+ \cdots +H_i)}{\crd{H_{i+1}}}  \geqslant \frac{\dstar}{2}
\end{align*}

Applying Lemma~\ref{lem.1} to the first term in the above chain and the induction statement for $i$, 
we obtain the desired result for $i+1$. 
Hence done.
\end{proof}

The next claim lower bounds the value $f(\dellhstar)$. 

Building up to the Claim, let us note that $\dellhstar \neq \emptyset$. 
If the intersection were empty, then $\hstar$ is a subgraph of density $\dstar$, and so $H_{\ell +1}$ would 
be a subgraph of density at least $\dstar$. But then,  
\[
\frac{f(D_\ell + H_{\ell + 1}) - f(D_\ell)}{\crd{H_{\ell + 1}}} \stackrel{supermodular}{\geqslant} \frac{f(H_{\ell + 1})}{\crd{H_{\ell + 1}}} \geqslant \dstar
\]
But this contradicts the choice of $D_\ell$.

\begin{claim}
\label{HHH}
\[
f(D_\ell \cap \hstar) \geqslant \frac{\dstar}{2}\crd{\dellhstar}+ \frac{\dstar}{2}\crd{\hstar}
\]
\end{claim}
\begin{proof}
Let $X = \hstar - \dellhstar$. Then, $X\cap D_\ell = \emptyset$ and $D_\ell + X = D_\ell \cup \hstar$. Then by definition of $D_\ell$, we know that 
$\frac{f(D_\ell + X) - f(D_\ell)}{\crd{X}} \leqslant \frac{f(D_{\ell+1}) - f(D_\ell)}{\crd{D_{\ell+1}} -\crd{D_\ell}} < \dstar/2$. Thus, 
$f(D_\ell \cup \hstar) - f(D_\ell) \leqslant \frac{\dstar}{2}(\crd{\hstar} - \crd{\dellhstar})$.

Therefore, 
$f(D_\ell \cup \hstar) + f(\dellhstar) \leqslant f(D_\ell) + f(\dellhstar) + \frac{\dstar}{2}(\crd{\hstar} - \crd{\dellhstar})$. 

Applying supermodularity we have that $f(D_\ell \cup \hstar) + f(\dellhstar) \geqslant f(D_\ell) + f(\hstar)$. 
Thus, cancelling $f(D_\ell)$ gives us that 
$f(\dellhstar) + \frac{\dstar}{2}(\crd{\hstar} - \crd{\dellhstar}) \geqslant f(\hstar)$. The claim follows by 
observing that $\dstar = \frac{f(\hstar)}{\crd{\hstar}}$.
\end{proof}

Note that this claim also implies that the density of the set $\dellhstar$ is at least $\dstar$. Intuitively, $\dellhstar$ is a subset that has 
``enough $f$-value'' as well as a ``good'' density. 

We may now combine the statements of Claim~\ref{GGG} and Claim~\ref{HHH} to get the following chain of inequalities:
\begin{align*}
f(D_\ell) \stackrel{Claim~\ref{GGG}}{\geqslant}
 f(\dellhstar) +\frac{\dstar}{2}\crd{D_\ell}  - \frac{\dstar}{2}\crd{\dellhstar} &\stackrel{Claim~\ref{HHH}}{\geqslant} 
\frac{\dstar}{2}\crd{D_\ell} +  \frac{\dstar}{2}\crd{\hstar}
\end{align*}

Consider $D'_\ell$:  this is obtained from $D_\ell$ by adding suitably many elements to make $D_\ell$ feasible. Let $r$ be the minimum {\em number} of elements to be added to $D_\ell$ so as to make it feasible. 
Since $\hstar$ is a feasible solution too, clearly, $r \leqslant \crd{\hstar}$. 
With this motivation, we define the {\em Extension Problem} for a matroid $\calM$. The input is a matroid $\calM = (U, \calI)$ and a subset $A \subseteq U$. The goal is to 
find a subset $T$ of minimum cardinality such that $\overline{A\cup T} \in \calI$. Lemma~\ref{extension} shows that we can find such a subset $T$ in polynomial time. 
Thus, we would have that:
\[
d(D'_\ell) = \frac{f(D'_\ell)}{\crd{D_\ell} + r} \geqslant \frac{f(D_\ell)}{\crd{D_\ell} + r} \geqslant \frac{f(D_\ell)}{\crd{D_\ell} + \crd{\hstar}} \geqslant \dstar/2
\]
and we are done with the proof of Theorem~\ref{main.thm}, modulo the proof of Lemma~\ref{extension}.

We proceed to present the lemma and its proof:
\begin{lemma}
\label{extension}
The Extension Problem for matroid $\calM$ and subset $A$ can be solved in polynomial time.
\end{lemma}
\begin{proof}
The proof considers the {\em base polyhedron} of the matroid (see the text by Schrijver~\cite{schrijver-book}). We will have a variable $x_i$ for each element $i \in U \setminus A$, where 
$x_i = 1$ would indicate that we pick the element $i$ in our solution $T$. For brevity, we will also 
maintain a variable $y_i$ that indicates whether $i$ is {\em absent} from the solution $T$. Thus 
for every $i$, we will maintain that $x_i + y_i = 1$. 
Given an arbitrary set $S$, we will let $r(S)$ denote the {\em rank} of the subset $S$ in the matroid $\calM$.

The following is a valid integer program for the Extension Problem (where $y(S)$ is shorthand for 
$\sum_{i \in S} y_i$). The linear program to the right is the relaxation of the integer program, and 
with variables $x_i$ eliminated. 
\[
\mathrm{IP_1:}
\begin{array}{cll}
\min &{\displaystyle \sum_{i \in U} x_i}\\[1.0em]
{\rm s.t.} &{\displaystyle 
x_i + y_i = 1}&{\rm for~all~}i \in U\\[0.25em]
&{\displaystyle y(S) \leqslant r(S)}&{\rm for~all~} S\subseteq U\\[0.25em]
&{\displaystyle x_i = 1}&{\rm for~all~} i \in A\\[0.25em]
&{\displaystyle
x_i, y_i \in \{0,1\}} & {\rm for~all~} i \in U\ .
\end{array}
\quad
\mathrm{LP_1:}
\begin{array}{cll}
\min &{\displaystyle \sum_{i \in U} (1 - y_i)}\\[1.0em]
{\rm s.t.} 
&{\displaystyle y(S) \leqslant r(S)}&{\rm for~all~} S\subseteq U\\[0.25em]
&{\displaystyle y_i = 0}&{\rm for~all~} i \in A\\[0.25em]
&{\displaystyle
y_i\geqslant 0} & {\rm for~all~} i \in U\ .
\end{array}
\]

The linear program $\lp_1$ can also be formulated as a maximization question. 
To be precise, let $\val(\lp_1)$ denote the {\em value} of the program $\lp_1$. Then 
$\val(\lp_1) = \crd{U} - \val(\lp_2)$, where $\lp_2$ is as follows:
\[
\mathrm{LP_2:}
\begin{array}{cll}
\max &{\displaystyle \sum_{i \in U} y_i}\\[1.0em]
{\rm s.t.} 
&{\displaystyle y(S) \leqslant r(S)}&{\rm for~all~} S\subseteq U\\[0.25em]
&{\displaystyle y_i = 0}&{\rm for~all~} i \in A\\[0.25em]
&{\displaystyle
y_i\geqslant 0} & {\rm for~all~} i \in U\ .
\end{array}
\]

Now, by folklore results in matroid theory (cf. \cite{schrijver-book}), we have that solutions to $\lp_2$ are integral and can be found by a greedy algorithm. Thus, we can solve $\ip_1$ in polynomial time, and this proves the statement of the Lemma. 
\end{proof}

\section{Proof of Theorem~\ref{knapsack}}
\label{app:knapsack}
In order to prove this result, we will have to modify the algorithm presented in Section~\ref{proofmain}. In the analysis, we will correspondingly modify the definition of the set $D_\ell$. Then we will apply (modified versions of) Claim~\ref{GGG} and Claim~\ref{HHH} to 
derive the result. 

The modified algorithm is as shown in Figure~\ref{knap-algo}.

\begin{figure}
\begin{center}
\begin{boxedminipage}{0.9\hsize}
\begin{algorithmic}
\State $i \gets 1$
\State $H_i \gets \argmax_{X} \frac{f(X)}{\crd{X}}$
\State $D_i \gets H_i$
\While{$D_i ~\mathrm{ infeasible}$}
\State $H_{i+1} \gets \argmax_{X: X\cap D_i = \emptyset} \frac{f(D_i + X) - f(D_i)}{\crd{X}}$
\State $D_{i+1} \gets D_i + H_{i+1}$
\State $i \gets i+1$
\EndWhile
\State $L \gets i$
\For{$i = 1 \to L$} 
\State{Order the vertices $i$ in $U\setminus D_i$ by non-increasing order of weights $w_i$}
\State{Add vertices from $U \setminus D_i$ in this order to $D_i$ until feasibility is attained}
\State{Let the result be $D'_i$}
\EndFor
\State {Output the subset among the $D'_i$'s with the highest density}
\end{algorithmic}
\end{boxedminipage}
\end{center}
\caption{Algorithm for Knapsack constraints}
\label{knap-algo}
\end{figure}

Consider the set $D_\ell$ in the 
sequence $D_1, D_2, \cdots, D_L$ such that 
the following hold:
\[
\frac{f(D_\ell) - f(D_{\ell - 1})}{\crd{D_\ell} - \crd{D_{\ell - 1}}} \geqslant \frac{\dstar}{3}
\]
but 
\[
\frac{f(D_{\ell+1}) - f(D_{\ell})}{\crd{D_{\ell+1}} - \crd{D_{\ell}}} < \frac{\dstar}{3}
\]
As earlier, if there is no such $\ell < L$ for which this holds, this implies that $D_L$ satisfies 
$\frac{f(D_L) - f(D_{L - 1})}{\crd{D_L} - \crd{D_{L - 1}}} \geqslant \dstar/3$. But this gives a $3$-approximation in this case since $D_L$ is feasible and 
\[
d(D_L) \geqslant \frac{f(D_L) - f(D_{L - 1})}{\crd{D_L} - \crd{D_{L - 1}}}
\]

Let us consider the other case where $\ell < L$ and $D_\ell$ is {\em infeasible}. Let $\hstar$ denote the optimal solution.

We state modified versions of Claim~\ref{GGG} and Claim~\ref{HHH}.
\begin{claim}(modified Claim~\ref{GGG})
\label{KKK}
\[
f(D_\ell) - f(\dellhstar) \geqslant \frac{\dstar}{3}(\crd{D_\ell} - \crd{\dellhstar})
\]
\end{claim}

\begin{claim}(modified Claim~\ref{HHH})
\[
f(D_\ell \cap \hstar) \geqslant \frac{\dstar}{3}\crd{\dellhstar}+ \frac{2\dstar}{3}\crd{\hstar}
\]
\end{claim}

These modified claims may be proven analogously to the original claims, taking into account the
new definition for $D_\ell$.

Now note that given a set $D_\ell$, in order to make the set feasible for the knapsack cover
constraint, we pick the elements with the largest weights $w_i$ so that feasibility is attained. 
The usual knapsack greedy algorithm shows that this is a $2$-approximation to the 
optimal knapsack cover. Thus, if we add $r$ elements, then $r \leq 2\hstar$. 
Thus we have that, 
\[
\frac{f(D'_\ell)}{D'_\ell} \geqslant \frac{f(D_\ell)}{D_\ell + r} \geqslant \frac{f(D_\ell)}{D_\ell + 2\hstar} 
\geqslant \frac{\dstar}{3}
\]
Thus we have proven a $3$-approximation.

\section{Proof of Theorem~\ref{depden}}
\label{sec:depden}
We will present the proof for the case of the graph density function, i.e. where $f(S) = \crd{E(S)}$. 
The proof for arbitrary $f$ will require a passage to the Lov\'{a}sz Extension ${\cL}_f(x)$ of a set function 
$f(S)$.
In fact we will present {\em two} proofs of this fact for the special case of the graph density function. 
To the best of our knowledge, both the proofs are new, and seems simpler than existing proofs. 
For both the proofs, we will use the same $\lp$.

\noindent
{\bf First Proof:}

We will {\em augment} the $\lp$ that Charikar~\cite{charikar00} uses to prove that graph density is computable in polynomial time. Given a graph $G = (V, E)$, there are edge variables $y_e$ and vertex variables $x_i$ in the $\lp$. We are also 
given an auxiliary {\em dependency} digraph $D = (V, \vec{A})$ on the vertex set $V$. In the augmented $\lp$, we also have
constraints $x_i \leqslant x_j$ if there is an arc from $i$ to $j$ in the digraph $D= (V, \vec{A})$.
The $\depden$ problem is modelled by the linear program $\lp_3$. 
\[
\mathrm{LP_3:}
\begin{array}{cll}
\max &{\displaystyle \sum_{e \in E} y_e }\\[1.0em]
{\rm s.t.} &{\displaystyle 
\sum_{i} x_i = 1}& \\[0.25em]
&{\displaystyle y_e \leqslant x_i}&{\rm for~all~} e \sim i, e \in E\\[0.25em]
&{\displaystyle x_i \leqslant x_j}&{\rm for~all~} (i, j) \in \vec{A}\\[0.25em]
&{\displaystyle
x_i\geqslant 0} & {\rm for~all~} i \in V(G)\ .
\end{array}
\quad
\mathrm{CP_1:}
\begin{array}{cll}
\max &{\displaystyle \sum_{e = (i,j)\in E} \min\{x_i, x_j\} }\\[1.0em]
{\rm s.t.} &{\displaystyle 
\sum_{i} x_i = 1}& \\[0.25em]
&{\displaystyle x_i \leqslant x_j}&{\rm for~all~} (i, j) \in \vec{A}\\[0.25em]
&{\displaystyle
x_i\geqslant 0} & {\rm for~all~} i \in V(G)\ .
\end{array}
\]

Suppose we are given an optimal solution $\hstar$ to the $\depden$ problem. 
Let $\val(\lp_3)$ denote the 
feasible value of this $\lp$: we will prove that $\val(\lp_3)  = d(\hstar)$.

\noindent
{\bf $\val(\lp_3)\geqslant d(\hstar)$:}

We let $\crd{\hstar} = \ell$, and $x_i = 1/\ell$ for $i \in \hstar$, and $0$ otherwise. 
Likewise, we set $y_e = 1/\ell$ for $e \in E(\hstar)$, and $0$ otherwise. Note that $\hstar$
is feasible, so if $a \in \hstar$ and $(a, b) \in \vec{A}$, then it also holds that $b\in \hstar$.
We may check that the assignment $x$ and $y$ is feasible for the $\lp$. So, $d(\hstar) = \frac{\crd{E(\hstar)}}{\ell}$ is achieved as the value of a feasible assignment to the $\lp$. 
 
\noindent
{\bf $\val(\lp_3)\leqslant d(\hstar)$:}

In the rest of the proof, we will prove that there exists a subgraph $H$ such that  $\mathrm{VAL} \leqslant d(H)$. 
First, it is easy to observe that in any optimal solution of the above $\lp$, the variables $y_e$ will take the values $\min\{x_i, x_j\}$ where $e = (i,j)$.
Thus, we may eliminate the variables $y_e$ from the  program $\lp_3$ to obtain the program $\cp_1$.
We claim that $\cp_1$ is a convex program. 
Given two concave functions, the $\min$ operator preserves concavity. Thus, the objective function of the above modified program is concave. Hence we have a convex program: here, the objective to be {\em maximized} is concave, subject to linear constraints. We may solve the program $\cp_1$ and get an output
 optimal solution ${x^{*}}$. Relabel the vertices of $V$ such that the following holds:
$x^{*}_1 \geqslant x^{*}_2 \geqslant \cdots \geqslant x^{*}_n$. If there are two vertices with (modified) indices $a$ and $b$ where $a < b$ and there is an arc $(a,b) \in \vec{A}$, then we have the equalities 
$x^{*}_a = x^{*}_{a+1} = \cdots = x^{*}_b$. We will replace the inequalities in the program $\cp_1$ as follows:
\[
\mathrm{LP_4:}
\begin{array}{cll}
\max &{\displaystyle \sum_{e = (i,j)\in E: i < j} x_j }\\[1.0em]
{\rm s.t.} &{\displaystyle 
\sum_{i} x_i = 1}& \\[0.25em]
&{\displaystyle x_i \geqslant x_{i+1}}&{\rm for~all~} i \in \{1, 2, \cdots, (n-1)\}\\[0.25em]
&{\displaystyle
x_n\geqslant 0} &\ .
\end{array}
\]

\noindent
where some of the inequalities $x_i \geqslant x_{i+1}$ may be equalities if there is an index $a$ with $a \leqslant i$ and an index $b$ with $b \geqslant (i+1)$ such that $(a, b) \in \vec{A}$. Note also that because of the ordering of the variables of this $\lp$, the objective function also simplifies and becomes a linear function. Clearly $x^{*}$ is a feasible solution to this $\lp$. Thus the value of this $\lp$ is no less than the value of $\cp_1$. Consider a $\bfs$ $x$ to $\lp_4$.  The program $\lp_4$ has $(n+1)$ constraints, 
and $n$ variables.
Given the $\bfs$ $x$, call a constraint {\em non-tight} if it does not hold with equality under the solution $x$. Thus, there may be at most one {\em non-tight} constraint in $\lp_4$.
In other words, there is at most one constraint $x_i \geqslant x_{i+1}$ that is a strict inequality. This, in turn, implies that all the non-zero values in $x$ are equal. Let there be $\ell$ such non-zero values. From the equality $\sum_i x_i = 1$, we get that each non-zero $x_i = 1/\ell$. Let $H$ denote the set of indices $i$ which have non-zero $x_i$ values. Then the objective value corresponding to this $\bfs$ $x$ is $\crd{E(H)}/\ell = d(H)$. 

Thus we have proven that $d(H) \geqslant \val(\lp_4) \geqslant \val(\cp_1) = \val(\lp_3)$, as required.
This completes the proof of Theorem~\ref{depden}.

\qed

\noindent
{\bf Remarks about the proof:}
\begin{itemize}
\item
We remark that the objective in the convex program $\cp_1$ is precisely the Lov\'{a}sz Extension ${\cL}_f(x)$ for the specific function $f = \crd{E(S)}$. Thus our proof shows that the $\lp$ provided by Charikar~\cite{charikar00} is precisely the Lov\'{a}sz Extension for the specific supermodular function $\crd{E(S)}$.
\item
Note that there are other proofs possible for this result. For instance, one can follow the basic argument of Charikar to show that $\lp_3$ satisfies $d(\hstar) = \val(\lp_3)$. The proof we provide above is new, and is inspired by the work of Iwata and Nagano~\cite{iwata-nagano09}. 
\item
Via our proof, we also prove that any $\bfs$ for the basic graph density $\lp$ has the property that all the
non-zero values are equal. This fact is not new: it was proven by Khuller and Saha~\cite{ks09} but we believe our proof of this fact is more transparent. 
\end{itemize}

\vspace{0.3in}

\noindent
{\bf Second Proof:}

Again, let us consider the program $\lp_3$. 

Similar to the above, it will suffice to prove that the $\bfs$ solutions to this $\lp$ have the property that all non-zero 
components are equal.

Consider the {\em constraint} matrix $B$ that consists of the LHS of the {\em non-trivial} constraints in the above $\lp$, 
without the constraint $\sum_i x_i = 1$. Thus $B$ consists of rows for the constraints $y_e \leqs x_i$ 
(for $e \in E: e \sim i$)
and the constraints $x_i \leqs x_j$ for $(i,j)\in \vec{A}$. 
The matrix $B$ is $\tum$: this is because it can easily be realised as the {\em incidence matrix} of 
a digraph. 

Thus the original constraint matrix consists of the matrix $B$ augmented by a single (non-trivial) constraint, consisting of the sum of the $x_i$'s being equal to $1$; and also the (trivial) nonnegativity constraints 
$x_i \geq 0$ and $y_e \geq 0$. Let $B'$ denote this augmented matrix. Note that 
$B'$ need not be $\tum$. 

Consider a basic feasible solution ($\bfs$) $\widetilde{v} = (y_1, \cdots, y_e, \cdots, y_m, x_1, \cdots, x_i, \cdots, x_n)$.
Since $\widetilde{v}$ consists of $(m+n)$ variables, there are $(m+n)$ constraints in the constraint matrix
$B'$ that are tight. Consider the submatrix $T$ formed by the tight constraints in the matrix $B'$. 
Since the constraint $\sum_i {x}_i = 1$ is always tight, this will be included as a row in the matrix $T$. 
Without loss of generality, let this row be the {\em last} row $r$ of $T$. Thus, $\widetilde{v}$ is the {\em unique} solution to the linear system $Tv = b$, where $b^{T} = (0, 0, \cdots, 0, 1)$. 

Note that, by previous considerations, the submatrix $T'$ of the matrix $T$ consisting of all the rows of $T$ but the {\em last} one, is $\tum$ (since $T'$ is then a submatrix of the matrix $B$).
The $s^{th}$ component (for $1 \leq s \leq (m+n)$) of $\widetilde{v}$ may be found by 
Cramer's rule as $\widetilde{v}_{s}  = \frac{\det(T_s)}{\det(T)}$, where $T_s$ is the matrix $T$ with the 
$s^{th}$ column {\em replaced} by the vector $b$. 

Note that $\det(T)$ is at most $|V(G)|= n$. This is because 
the row $r$ has at exactly $n$ $1$'s, so we may expand the determinant by row $r$. Any sub-determinant
to be computed in this row-wise expansion of the determinant is a submatrix of $T$, thus is 
$\tum$. Therefore, $\det(T)$ is a sum of at most $n$ $+1$'s and $-1$'s thus, is (say) $k$ where $k \leq n$.

Consider the computation of $\det(T_s)$. The matrix $T_s$ has its $s^{th}$ column replaced by the vector $b$,
which has precisely one $1$. So we may expand the determinant of $T_s$ by its $s^{th}$ column, and thereby, 
the determinant is that of a square submatrix of the matrix $T$. This means that $\det(T_s)$ is $0$, $1$ or 
$-1$.

Thus, every component of $\widetilde{v}$ is precisely $0$ or $\frac{1}{k}$, or $-\frac{1}{k}$. However every component in the $\lp$ is $\geq 0$, thus the third possibility is excluded. This completes the proof.
\qed

\subsection{Arbitrary monotone supermodular functions:}

We now proceed to consider the case 
where we are given an arbitrary monotone supermodular function $f$ over the universe $U$ and a directed graph 
$D = (U, \vec{A})$, where the arcs in $\vec{A}$ specify the dependencies. 

To extend our results to this, we will need the concept of the Lov\'{a}sz Extension.

The Lov\'{a}sz Extension ${\cL}_f : [0,1]^{U} \rightarrow \mathbb{R}$, first defined by Lov\'{a}sz,
is an extension of an arbitrary set function $f : 2^U \rightarrow \mathbb{R}$. 
We proceed with the formal definition:

\begin{definition}(Lov\'{a}sz Extension)
Fix $x \in [0,1]^U$, and let $U = \{v_1, v_2, \cdots, v_n\}$ be such that $x(v_1) \geq x(v_2) \geq \cdots \geq x(v_n)$. For $0 \leq i \leq n$, let $S_i = \{ v_1, v_2, \cdots, v_i\}$. 
Let $\{\lambda_i\}_{i=0}^{n}$ be the unique coefficients with $\lambda_i \geq 0$, and 
$\sum_i \lambda_i = 1$ such that:
\[
x = \sum_{i=0}^{n} \lambda_i{1}_{S_i}
\]
It is easy to see that $\lambda_n = x(v_n)$, and for $0\leq i < n$, we have 
$\lambda_i = x(v_i) - x(v_{i+1})$ and $\lambda_0 = 1 - x(v_1)$. The {\em value} of the 
Lov\'{a}sz Extension of $f$ at $x$ is defined as 
\[
\cL_f(x) = \sum_i \lambda_i f(S_i)
\]
\end{definition}

For motivation behind the definition, refer to the excellent survey on submodular functions by 
Dughmi~\cite{dughmi09}. 

The Lov\'{a}sz Extension enjoys the following properties:
\begin{itemize}
\item $\cL_f$ is concave iff $f$ is supermodular.
\item If $f$ is supermodular, the maximum value of $f(S)$ is the same as the maximum value of 
$\cL_f(x)$.
\item Restricted to the subspace $x_1 \geq x_2 \geq \cdots \geq x_n$, 
the function $\cL_f$ is {\em linear}.
\end{itemize}

We are now ready to describe our convex program $\mathrm{CP}$ for computing the densest subset of the universe $U$  subject to {\em dependency constraints}. For details on convex programming, one may consult the text~\cite{bv-book}.

The program has variables $x_1, x_2, \cdots, x_n$ corresponding to the elements $i \in U$.
Since $f(S)$
is supermodular, the corresponding Lov\'{a}sz Extension $\cL_{E}(x)$ is {\em concave}.

\begin{center}
\[
\mathrm{CP:}
\begin{array}{cll}
\max &{\displaystyle \cL(x) }\\[1.0em]
{\rm s.t.} &{\displaystyle 
\langle x, \vec{1}\rangle = 1}& \\[0.25em]
&{\displaystyle x_i \leqslant x_j}&{\rm for~all~} (i,j) \in \vec{A}\\[0.25em]
&{\displaystyle
x_i\geqslant 0} & {\rm for~all~} i \in V(G)\ .
\end{array}
\]
\end{center}

This convex programming problem can be solved to arbitrary precision in polynomial time by 
the ellipsoid method (see \cite{gls81}). 

As in the first proof above, we will relabel the elements of the universe so that 
$x^{*}_1 \geqslant x^{*}_2 \geqslant \cdots \geqslant x^{*}_n$. But now, by the property of the 
Lov\'{a}sz Extension, we see that $\cL(x)$ is a {\em linear} function in this subspace. 

Now, the rest of the first proof carries over and gives us the result for 
arbitrary monotone supermodular $f$. 

\subsection{Proof of Corollary~\ref{BBB}}
\label{cor.proof}
There are many ways to see this. One way is to consider the convex program above
for the Lov\'{a}sz Extension of the monotone supermodular function $f$.

\[
\mathrm{CP_2:}
\begin{array}{cll}
\max &{\displaystyle \cL(x) }\\[1.0em]
{\rm s.t.} &{\displaystyle 
\langle x, \vec{1}\rangle = 1}& \\[0.25em]
&{\displaystyle x_i \geqslant x_{i+1}}&{\rm for~all~} i \in \{1, 2, \cdots, (n-1)\}\\[0.25em]
&{\displaystyle
x_n\geqslant 0} &\ .
\end{array}
\]
As in the proof above, we can see that this has solutions $x^{*}$ where all the nonzero $x_i$'s are equal, 
and that this corresponds to choosing a subset $S$ so that $\cL(x^{*}) = f(S)/\crd{S}$. Thus, we see that 
$\max_S f(S)/\crd{S}$ is computable in polynomial time.

Yet another way of verifying Corollary~\ref{BBB} is to consider the sequence of  functions
$g(\alpha, S) \triangleq f(S) - \alpha\crd{S}$ (for fixed $\alpha \geqs 0$). Note that each $g(\alpha, S)$ is 
supermodular for any fixed $\alpha$, and so can be maximized in polynomial time. Also observe that if 
$\max_S f(S)/\crd{S} \geqs \alpha$ for some $\alpha$, then $\max_S g(\alpha, S) \geqs 0$. Conversely, 
if $\max_S f(S)/\crd{S} \leqs \alpha$ then $\max_S g(\alpha, S) \leqs 0$. Thus, we can find 
$\max_S f(S)/\crd{S}$ by a binary search over $\alpha$ and maximizing the corresponding functions $g(\alpha,S)$. 

\section{Proof of Theorem~\ref{combo}}
\label{app:combo}
To fix the notation, in this problem, we are given a monotone supermodular function $f$ over a universe $U$, a matroid $\calM = (U, \calI)$, and a set $A \subseteq U$. 

The only modification that we have to make to Algorithm~\ref{algo} is that we will choose the first set $H_1$ such that $H_1$ maximizes the density of all subsets that contain the set $A$. Note that we can do this in polynomial time. Apart from this, the construction of the sets $H_2, H_3, \cdots, H_L$ and the sets $D_1, D_2, \cdots, D_L$ are the same as in Algorithm~\ref{algo}. So, each $D_i$ contains $A$, and the candidate feasible solutions $D'_i$ also contain $A$. The analysis of the modified algorithm is the same as in Section~\ref{proofmain}. Thus we obtain a $2$-approximation algorithm as promised in Theorem~\ref{combo}.

\section{Open Problems}
One interesting open direction is to investigate the maximization of density functions subject to combinations of constraints. In this paper, we consider the combination of a single matroid and a subset constraint. In general, one could ask similar questions about combinations of multiple matroid constraints or a matroid and a dependency constraint for instance. Another open question is to derive a $\lp$-based technique to prove the  result in Theorem~\ref{main.thm}.
\bibliographystyle{plain}
\bibliography{fst}
\appendix

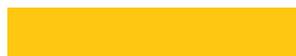
\end{document}